\documentclass[11pt]{article}

\usepackage{fullpage}
\usepackage{amssymb,amsmath}
\usepackage{graphicx, epsfig}

\def\idrm#1{\ensuremath{\mathrm{#1}}}

\def\floor#1{\lfloor #1 \rfloor}
\def\ceil#1{\lceil #1 \rceil}

\newtheorem{theorem}{Theorem}
\newtheorem{lemma}{Lemma}

\newtheorem{corollary}{Corollary}
\newtheorem{problem}{Problem}
\newenvironment{proof}{\trivlist\item[]\emph{Proof}:}%
{\unskip\nobreak\hskip 1em plus 1fil\nobreak$\Box$
\parfillskip=0pt%
\endtrivlist}

\newenvironment{itemize*}%
  {\begin{itemize}%
    \setlength{\itemsep}{0pt}%
    \setlength{\parskip}{0pt}%
    \setlength{\parsep}{0pt}%
    \setlength{\topsep}{0pt}%
    \setlength{\partopsep}{0pt}%
  }%
  {\end{itemize}}%

\newcommand{\halfleftsect}[2]{(#1,#2]}
\newcommand{\cL}{{\cal L}}
\newcommand{\cA}{{\cal A}}
\newcommand{\cT}{{\cal T}}
\newcommand{\cP}{{\cal P}}
\newcommand{\cD}{{\cal D}}
\newcommand{\cM}{{\cal M}}
\newcommand{\cB}{{\cal B}}

\newcommand{\todo}[1]{ }

\newcommand{\prank}{\idrm{prank}}
\newcommand{\docc}{\idrm{docc}}
\newcommand{\freq}{\idrm{freq}}
\newcommand{\mindist}{\idrm{mindist}}
\newcommand{\seq}{\idrm{seq}}
\newcommand{\rel}{\idrm{rel}}

\newcommand{\eps}{\varepsilon}

\begin{document}

\title{\bf Top-$K$ Color Queries for Document Retrieval}
\author{
Marek Karpinski\thanks{Department of Computer Science, 
University of Bonn. Supported in part by DFG grants and the Hausdorff Center 
excellence grant EXC 59-1. 
Email {\tt marek@cs.uni-bonn.de}.}
\and
Yakov Nekrich\thanks{Department of Computer Science, 
University of Bonn. Supported in part by the Hausdorff Center excellence
 grant EXC 59-1. 
Email {\tt yasha@cs.uni-bonn.de}.}
}
\date{}
\maketitle
\begin{abstract}
In this paper we describe a new efficient (in fact optimal) 
 data structure for the  {\em top-$K$ color problem}.
Each element of an array $A$ is assigned a color $c$ with priority $p(c)$. 
For a  query range $[a,b]$ and  a value $K$, we have to report $K$ 
 colors with the highest priorities among all colors 
that occur in  $A[a..b]$, sorted in reverse order by their priorities.
We show that such queries can be answered in $O(K)$ time using an 
$O(N\log \sigma )$ bits data structure, where $N$ is the number of elements in 
the array and $\sigma$ is the number of colors. 
Thus our data structure is asymptotically optimal with respect to  the 
worst-case 
 query time and  space.  
As an immediate application of our results, we obtain optimal time 
solutions 
for several document retrieval problems. The method of the paper could be also
 of independent interest. 

\end{abstract}


\todo{ Bad notation: $P$ -pattern and sets $P(j_1,j_2)$, $d$ document and 
the number of descendants, and $d$-dimensional points in the Intro, $p(c)$ - priority and $p$ as index}
\todo{sorted and unsorted queries- define or remove where they are mentioned} 
 \section{Introduction}
 In this paper we study a variant of the well-known  color reporting problem.
 Each entry of an array $A$ is assigned a color $c\in C$ with priority $p(c)$. 
 For a  query range $Q=[a,b]$ and an integer $K$, the data structure 
 reports $K$ distinct colors with  highest priorities among all colors that 
 occur in $Q$. Colors are reported in the reverse order of their priorities. 
 In the online version of this problem, we report all colors that occur 
 in $A[a..b]$ in decreasing order until all colors are reported or the 
 query is terminated by the user. 

 Using an $O(N\log\sigma)$ bits data structure, we can answer such queries 
 in $O(K)$ time, where $K$ is the number of reported colors and 
 $\sigma$ is the total number of distinct colors in $A$. Thus our 
 data structure achieves worst-case optimal query time and space usage.
 Even for a simpler problem of reporting all distinct colors  in $A[a..b]$ 
 in arbitrary order, 
 the best previously known optimal time data structure  uses $O(N\log N)$ bits.

 The study of this problem is motivated by its applications to document 
 retrieval and search engines. 
It is known~\cite{M02} that we can report  all documents that 
 contain a pattern $P$ by reporting all distinct colors that occur 
 in a range $A[a_p..b_p]$ of the document array. 
In many cases, we want to output 
only most important or most relevant documents in sorted order 
starting with the most important (relevant) documents.
The well known example of such scenario are search engines: an answer 
to a query is a sequence of documents output in the reverse order of their 
relevance. 
Static  ranking of documents based on 
 e.g. their links with other documents, such as PageRank~\cite{PBMW99} and 
 HITS~\cite{Kl99} is an important part of estimating document relevance. 
Thus it may be 
beneficial to generate the list of $K$ most highly ranked 
documents that contain a specified 
pattern $P$,  sorted by their ranks. 
The parameter $K$ is sometimes not known in advance and documents 
must be reported in order of their ranks until the procedure is terminated. 
%
In such situations our data structure gives us an optimal time 
solution.
Our result can  be also applied to other document retrieval problems.



   
 {\bf Previous and Related Work.}
 Colored range searching is a widely studied problem with various applications. 
 In computational geometry and data structures, 
 the following variant of the problem is considered. 
 A set of   points is stored in a data structure, 
 so that for any 
 rectangle $Q$ distinct colors of all points in $Q$ must be reported  
or the number 
 of distinct colors must be counted.   
 Such queries can be supported efficiently for $l \leq 3$ 
 dimensions~\cite{JL93,GJS95,AGM02}. 
 Several related problems, in which distinct colors of geometric objects 
 must be reported or counted were also studied.

 In the document listing problem, a set of documents $d_1,\ldots, d_s$ 
 with total length $N$ must be stored in a data structure, so that 
 for any pattern $P$ all documents that contain $P$ must be reported.  
 The total number of occurrences of $P$ may significantly exceed 
 the number of documents that contain $P$.
 Matias et al.~\cite{MMSZ98} described the first data structure 
 for this problems; their data structure  answers document listing  queries 
 in $O(|P|\log s +\docc)$ time, where $|P|$ is the length of $P$ and
 $\docc$ is the number of  reported documents.
 Muthukrishnan~\cite{M02} showed that several document retrieval problems 
 can be reduced to colored searching problems. In~\cite{M02} the author 
describes 
 an $O(N\log N)$ bits data structure that answers  document listing  query 
  in optimal $O(|P|+\docc)$  time. 
The data structures of~\cite{S07,VM07} further improve the space usage 
by storing the documents in compressed form; however, their solutions 
do not achieve optimal query time: it takes $O(\log^{\eps} N)$ time~\cite{S07}
or $O(\log s)$ time to report each document. 
The solution of Gagie et al.~\cite{GPT09}, based on the wavelet tree, uses 
$N\log s$ bits but also needs suboptimal $(|P|+\docc\log s)$ bits to 
answer a query.

The total number of documents that contain $P$ can  be very large and we 
may be interested in reporting only a subset of documents that contain 
the pattern $P$. In~\cite{M02}, two such problems are considered. 
In the $K$-mine problem, documents that contain at least $K$ occurrences 
of pattern $P$ must be reported. 
In the $K$-repeats two problems, we report all documents $d$, such that 
the minimal distance between two occurrences of the pattern $P$ in $d$ 
is at most $K$.  
In~\cite{M02}, $O(N\log^2 N)$ bit data structures that solve both problems 
in $O(|P|+\docc)$ time are described. 

Instead of reporting all documents whose relevance score exceeds a certain 
threshold, we often want to report $K$ most important or 
most relevant documents in sorted order. 
Recently, Hon et. al~\cite{HSV09} addressed this problem 
and described an efficient framework for 
reporting the $K$ most relevant documents with respect to 
the query pattern $P$. 
Their data structure uses linear space (i.e., $O(N\log N)$ bits) and can 
report $K$ most relevant 
documents in $O(|P|+K\log K)$ time. They also describe a  compressed 
data structure 
that  supports queries in $O(|P|+K\text{polylog} (N))$ 
time. In addition to static document ranks, the framework of~\cite{HSV09} 
also supports  other  relevance metrics


The problem of storing an array $A$, so that for any $a<b$ 
all elements in $A[a]A[a+1]\ldots A[b]$ can be output in sorted order 
was studied by Brodal et al~\cite{BFGO09}. In~\cite{BFGO09} the authors 
obtained an $O(N \log N)$ bits and optimal $O(|b-a+1|)$ time solution for 
this problem. 
We observe that in this paper a different problem is studied: 
if array $A$ contains colors and some color $c$ occurs $n_c$ times 
in $A[a]A[a+1]\ldots A[b]$ then the data structure of~\cite{BFGO09} 
reports this color $n_c$ times. Our data structure returns the color 
$c$ only once in this situation. 
The problem of ranked reporting was also considered by Grossi and
 Bialynicka-Birula~\cite{BG07}. They describe a general technique for 
adding rank information to geometric objects so that  answers to 
range reporting queries can be ordered by rank. However, any data structure  
based on their method uses super-linear 
space and 
requires  poly-logarithmic time to answer queries. 
For instance, a reduction of color queries to three-sided queries 
and their  method result in an $O(N\log^{2+\eps}N)$ bit data structure 
that answers  queries in $O(\log^2 N + K)$ time.

{\bf Our Results.}
We develop a new explicit technique for recursive, exponentially 
decreasing size subarrays combined with a new method for storing 
certain, pre-defined query answers. 
 We show that an array $A$ can be 
stored in an $O(N\log \sigma)$ bits data structure so that 
for any two indexes $a<b$ and for any integer $K$, 
 $K$ distinct colors with  highest priorities among all colors that 
 occur in $A[a..b]$  can be reported in optimal $O(K)$ time. 
In fact, it is not necessary to know  $K$ in advance: we can 
report colors that occur in $A[a..b]$ in the reverse order of their 
priorities until all colors are reported or the  procedure 
is terminated by the user.  Our method depends on transforming 
a data structure with $O(N^{1/f}+K)$, $f>1$, query time into 
a data structure with optimal query time; two crucial components 
of this transformation are an efficient method for obtaining 
solutions for pre-defined intervals and recursively 
defined  data structures with exponentially decreasing number of 
elements. 

Our data structure leads to optimal time solutions 
 for document listing 
in situations when every document is assigned  a static rank.
\begin{problem}[Ranked Document Listing Problem]
Documents $d_1,\ldots, d_s$ are stored in a data structure, so that 
for any pattern $P$ and any $K$ we must return $K$ most highly ranked 
documents that contain $P$ ordered by their rank.
\end{problem}
The data structure of~\cite{HSV09} uses $O(N\log N)$ bits and solves 
this problem in $O(|P|+K\log K)$ time, 
where $N$ is the total length of all documents.
The compressed data structure of~\cite{HSV09} uses  
$CSA|+o(N)+s\log (N/s)$ bits, but requires $O(|P|+K\log^{3+\eps} N)$ time 
to answer a query, where $|CSA|$ denotes the number of bits necessary
 to store compressed suffix array for all documents.
We can solve the ranked document listing problem in optimal $O(|P|+K)$ time 
using  worst-case optimal $O(N\log s)$ bits of space (in addition to the 
suffix array). 
Even for the general document listing problem, the previous optimal time 
data structure~\cite{M02} needs $O(N\log N)$ additional bits of space.

\begin{problem}[Ranked $t$-Mine Problem]
Documents $d_1,\ldots, d_s$ are stored in a data structure, so that 
for any pattern $P$ and any $K$, $t$ we must return $K$ most highly ranked 
documents that contain $P$ at least $t$ times ordered by their rank.
\end{problem}
We can solve the ranked $t$-mine problem in $O(|P|+K)$ time by a data 
structure that uses   $O(N\log s)$ words of $\log N$ bits. 


We can also combine our data structure with the framework of~\cite{HSV09}
and use a number of other relevance metrics. Let $S(d,P)$ denote the set of 
all positions in a document $d$, where $P$ matches. The framework 
of~\cite{HSV09} supports relevance metrics that depend on $S(d,P)$.
We will denote such a metric by $\rel(d,P)$.  
Examples of  metrics $\rel(d,P)$ 
are $\freq$, the frequency of occurrence of $P$ 
in a document, and  $\mindist$, the minimal distance between two occurrences 
of $P$ in a document.  
\begin{problem}[Most Relevant Documents Problem]
Documents $d_1,\ldots, d_s$ are stored in a data structure, so that 
for any pattern $P$ and any $K$ we must return $K$ most relevant  
documents with respect to a metric $\rel(d,P)$ ordered by $\rel(d,P)$.
\end{problem} 
The $O(N\log N)$ bit data structure of~\cite{HSV09} supports 
most relevant documents queries in $O(|P|+K\log K)$ time. 
For some relevance metrics, 
the compressed data structure of~\cite{HSV09} uses $2|CSA|+ s\log(N/s)+o(N)$ 
bits, but needs $O(|P|+K\text{polylog} (N))$ time to answer queries. 
For instance, if $\freq$ is chosen as the relevance metrics, then 
queries can be answered in $O(|P|+K\log^{4+\eps} N)$ time

We show that using a linear space data structure, 
we can report $K$ documents that 
contain a pattern $P$ and are most relevant with respect to $P$ 
in $O(|P|+K\log |P| )$ time. This is an improvement over the first
 result of~\cite{HSV09} for the case when $|P|=o(K)$. 
Moreover, if $|P|=\log^{O(1)} N$, then our data structure supports most 
relevant document queries in optimal  $O(|P|+K)$ time. 
For instance, suppose that  $\freq$ is used as relevance metric. 
Then for any pattern $P$ such that  $|P|=\log^{O(1)} N$, 
we can report $K$ documents in which  $P$ occurs most frequently in 
optimal $O(|P|+K)$ time. 



\todo{
We can also extend this

 For simplicity, we sometimes do not distinguish between elements and their 
 colors.

 A block consists of $B$ words. In this paper we also assume that 
 each word consists of $\log N$ bits. 
 Sometimes we will specify the space usage of our data structures in bits. 
 We observe, however, that if we describe an external memory 
  data structure that uses 
 $O(s(N))$ bits, then this data structure 
 can be packed into $O((s(n)/(B\log N)))$ block of space. 

 We distinguish between two variants of the top-$K$ color problem. 
 In the sorted top-$K$ color problem, $K$ colors with highest priorities 
 must be reported in the descending order of their priorities.
 In the unsorted top-$K$ color problem, $K$ colors with highest priorities 
 are reported in arbitrary order. 
 In this model we obtain the following results:\\
 (i) There exists an external memory data structure that uses 
 $O((N/B)\log\log N)$ blocks of space and answers unsorted top-$K$ queries 
 with  $O(K/B)$ I/O operations.\\
 (ii) There exists an external memory data structure that uses 
 $O((N/B)\log\log N)$ blocks of space and answers sorted top-$K$ queries 
 with $O((K/B)\log_B K)$ I/O operations.\\
}

{\bf Overview} In section~\ref{sec:colrep} we recall the results 
for standard one-dimensional color reporting and counting problems. 
In section~\ref{sec:nlognspace} we present a simple data structure 
that finds the (unsorted) 
list of $K$ colors with highest priorities in $O(K+\log^2 N)$ time
and uses $O(N\log^2 N)$ bits of space. 
Essentially, our data structure is a wavelet tree~\cite{GGV03} with 
secondary structures for color reporting and counting stored in its nodes.

In sections~\ref{sec:linspace} and~\ref{sec:opttime} 
we describe a new approach that 
enables us to achieve optimal time and almost optimal ($O(N\log N)$ bits) 
space. Our first idea, described in section~\ref{sec:linspace}, 
is sparsification: we store data structures only 
for nodes situated on a constant number of levels. This allows us to achieve 
linear space because each element is stored in only a constant number of data 
structures. On the other hand, our new
 search procedure must visit a much larger 
number of nodes; therefore, the search time grows to $O(N^{1/f}+K)$ for 
a constant $f$.  In section~\ref{sec:linspace} we show how the search time 
can be decreased without increasing space. 
First, we describe how we can obtain solutions for some pre-defined queries
using linear space. We recursively combine this method with data structures of 
section~\ref{sec:linspace}.
In section~\ref{sec:optspace} we demonstrate that the space usage can be 
further reduced to $O(N\log \sigma)$ bits; see Table~\ref{tbl:over}
Besides that, our data structure can be also extended to the external memory 
model, as shown in section~\ref{sec:extmem}.
Applications of our data structures to document retrieval problems 
are described in section~\ref{sec:docret}. 

\begin{table}
\centering
\begin{tabular}{|l|c|c|} \hline
                             & Query Time & Space Usage (in bits) \\ \hline
section~\ref{sec:nlognspace} & $O(\log^2 N + K)$ & $O(N\log^2 N)$ \\
section\ref{sec:linspace}    & $O(N^{1/f}+K)$  & $O(N\log N)$ \\
section~\ref{sec:opttime}    & $O(K)$  & $O(N\log N)$ \\
section~\ref{sec:optspace}   & $O(K)$  & $O(N\log \sigma)$\\ \hline
\end{tabular}
\caption{Overview of results in the RAM model}
\label{tbl:over}
\end{table}

Throughout this paper $A[i..j]$ denotes the subarray that 
consists of elements $A[i]A[i+1]\ldots A[j]$; $[a,b]$ denotes an interval 
that consists of all integers $x$, $a\leq x\leq b$. For simplicity, we 
sometimes do not distinguish between elements and their  colors.
Our data structures use only additions, subtractions, and standard 
bit operations. We say that a data structure with $N$ elements uses 
linear space if it can be stored in $O(N\log N)$ bits.

\section{Colored Reporting and Counting}
\label{sec:colrep}
In the color reporting problem, each element of an array $A$ is assigned 
a color $c$ from the set of colors $C$. Given a query range $[a,b]$, 
we must report all distinct colors $c_1,\ldots c_K$, such that at least 
one element colored with $c_i$, $1\leq i \leq K$, occurs in $A[a..b]$. 
In the color counting problem, we must count the number of distinct colors 
that occur in $A[a..b]$. Both problems were studied extensively; we 
refer the reader to~\cite{GJS95} for a survey of  results\footnote{Definitions 
of colored reporting and counting used in this paper are slightly more 
restrictive than the standard definitions of this problem.}.   
\begin{lemma}\label{lemma:ram}
In the RAM model, colored range reporting queries can be answered  in $O(K)$ time using an $O(N\log N)$ bits data structure. In the RAM model, the colored 
range counting problem can be solved in $O(\log N)$ 
time using an $O(N\log N)$ bits data structure.  
\end{lemma}
\begin{proof}
As shown in~\cite{GJS95}, the one-dimensional colored reporting (counting) 
for an array with $N$ elements can be reduced to the standard 
three-sided reporting (resp.\ counting) on $N\times N$ grid, 
i.e. to the problem of storing a set of two-dimensional points
whose coordinates belong to an integer interval $[1,N]$ in a data structure,
so that all points that belong to 
a query range of the form $[x_1,x_2]\times \halfleftsect{y_1}{+\infty}$
can be reported (counted).
Three-sided reporting queries on  the $N\times N$ grid 
can be answered in $O(K)$ time in the
 RAM model using an $O(N\log N)$ bits data structure~\cite{ABR00,M02}. 
Three-sided counting queries can be answered in $O(\log N)$ time using 
a linear space data structure~\cite{Ch88}.
\end{proof}
\begin{lemma}~\label{lemma:ext}
In the external memory model, colored range reporting queries can be answered 
 in $O(\log\log_B N + K/B)$ I/Os using an $O(N\log N)$ bits data 
structure. In the external memory  model, 
the colored range counting problem can be solved in $O(\log N)$ 
I/Os  using an $O(N\log N)$ bits data structure.  
\end{lemma}
\begin{proof}
We use the same reduction to three-sided reporting (counting) as in 
Lemma~\ref{lemma:ram}. There exists an $O(N\log N)$ bits data structure that 
supports  three-sided reporting queries in $O(\log \log_B n + K/B)$ I/O 
operations~\cite{N07}. The result of~\cite{Ch88} can be straightforwardly 
extended to the external memory model.
\end{proof}

\section{An $O(N\log N)$ Space Data Structure}
\label{sec:nlognspace}



In this section we consider a simple problem: $K$ colors with highest 
priorities that occur in 
the query interval $[a,b]$ must be reported in an  arbitrary order.
The data structure described in this section is based on recursive partitioning of the set of colors based on their priorities. Thus our approach in this 
section is similar to  the idea of the 
wavelet tree. 
Every node of a binary tree $T$ is associated with a set of colors $C_v$ and
an array $A_v$. If $v$ is the root node, then $A_v=A$ and $C_v=C$. 
When $A_v$ and $C_v$ for some node $v$ of $T$ are known, the arrays 
for the children of $v$ can be constructed.  
The set of colors $C_v$ is divided into two sets $C_0$ and $C_1$ 
that contain equal number of elements\footnote{We assume that $\sigma=|C|$ is a power of two.} and all colors in $C_0$ have  smaller 
priorities than any color in $C_1$.  
We denote by $N_v$ the total number of elements in $A_v$.
We store an additional array $\cB_v$ of 
$N_v$ bits; 
the $i$-th bit of $\cB_v$ equals to $1$ if and only if the color $A_v[i]$ belongs 
to $C_1$. If $u$ and $w$ are the right and left children of $v$, then 
we set $C_w=C_0$ and $C_u=C_1$. The array  $A_u$ ($A_w$) contains all 
elements of $A_v$ whose colors belong to $C_u$ ($C_w$): if $A_v[i]$ belongs 
to $C_u$ and there are $l_i$ indexes $j$, such that $\cB_v[j]=0$ and $j\leq i$, 
then $A_v[i]$ is stored at position $l_i$ in $A_u$; if $A_v[i]$ belongs 
to $C_w$ and there are $l_i$ indexes $j$, such that $\cB_v[j]=1$ and $j\leq i$, 
then $A_v[i]$ is stored at position $l_i$ in $A_w$.  
Every array $\cB_v$ is augmented with the rank/select data structure that 
enables us to count the number of $1$'s or $0$'s that occur in 
$\cB_v[1..j]$ for any $j\leq N_v$. 
Using $\cB_v$ we can count the number of elements in $A_v[1..j]$ whose 
colors belong to $C_u$ or $C_w$.

Furthermore, we also store  data structures $COUNT_v$ and $REP_v$ in each 
node $v$. The data structures $COUNT_v$  and 
$REP_v$ support color counting and color reporting queries on $A_v$.
A tree $T$ with arrays $\cB_v$ is the standard wavelet tree. Thus our 
construction can be viewed as a wavelet tree with auxiliary data structures 
for color reporting and color counting stored at its nodes.
The height of $T$ is $\log\sigma\leq \log N$; hence, every element is stored 
in $\log\sigma$ secondary data structures.

We will say that the interval $[a_v,b_v]$ corresponds to an interval 
$[a,b]$ in a node $v$ if all elements of $A[a..b]$ that belong to $C_v$ 
appear in $A_v[a_v.. b_v]$ in the same order as in $A[a..b]$. 
If we know $a_v$ and $b_v$ for a node $v$, then $a_u$ and $b_u$ for the 
right child $u$ of $v$ can be found using $\cB_v$. We set $a_u$ to the number 
of $1$'s in  $A_v[1..a_v]$; if $\cB[a_v]=0$, then $a_u$ is incremented by $1$. 
We set $b_v$ to the number of $1$'s in  $A_v[1..b_v]$. 
Values of $a_w$ and $b_w$ for the left child $w$ can be found in a symmetric 
way.

We can report the top $K$ colors in the interval $[a,b]$ 
using the algorithm that visits the sequence of nodes starting 
at the root of $T$. In every visited node we proceed as follows. 
Initially, we set $a_v=a$ and  $b_v=b$  for the root node $v$.\\
1. We use $\cB_v$ to find $[a_u,b_u]$ that corresponds to $[a,b]$ 
in  the right child $u$ 
of $v$. \\ 
2. We visit the node $u$ and count the number $m_u$ of distinct colors 
in $A_u[a_u..b_u]$. Obviously, $m_u$ equals to the number of distinct 
colors from  $C_u$ that occur in $A[a..b]$. \\
3. If $m_u\geq K$, we report the top $K$ colors in $A_u[a_u..b_u]$
using the same procedure (i.e., we set $v=u$ and return to step 1). \\
4. If $m_u< K$, we report all colors in $A_u[a_u..b_u]$ using $REP_u$.\\ 
5. Then, we report the top $K-m_u$ colors in the left child $w$ of $v$:
We use $\cB_v$ to find $a_w$ and $b_w$, where $[a_w,b_w]$ 
corresponds to $[a,b]$ in $w$.
Then, we set $K=K-m_u$, $v=w$, and return to step 1. \\
The total number of visited nodes is $O(\log N)$. In every node
we answer at most one color counting query and at most one color reporting 
query. By Lemma~\ref{lemma:ram}, 
the counting query can be answered in $O(\log N)$ time.
The color reporting query in a node $v$ can be answered 
in $O( K'_v)$ time, 
where $K'_v$ denotes the number of colors reported in the node $v$.
Thus a query can be answered in $O(\log^2 N  + K)$ time.
\begin{lemma}\label{lemma:nlogn}
There exists an $O(N\log^2 N)$ bits data structure that outputs an unsorted 
list of  top-$K$ colors in a query interval $[a,b]$ in $O(\log^2 N + K)$ time. 
\end{lemma}

\section{A Linear  Space Data Structure}
\label{sec:linspace}
In this section we describe a linear space  data structure that enables us to
 answer queries in $O(N^{1/f} + K)$ time. Our main idea is to store 
reporting and counting data structures only in selected nodes of the tree $T$,
so that each element is stored only in a constant number of data structures. 

We say that a  node $v$ is on level $x$ if $v$ has $x$ ancestors.
We say that a node $v$ is an \emph{important} node if $v$ 
is situated on  $i\floor{(1/f)\log N}$-th level for $i=0,1,\ldots, f$ 
and a constant $f$,  or $v$ is a leaf node. 
Instead of storing the array $A_v$ and the auxiliary data 
structures $REP_v$ and $CNT_v$ in every node $v$ of $T$, we only store them
in the important nodes. 
Besides that, we also store a data structure $E_v$ in every important node.
Let $v_1,\ldots, v_{t}$ be the highest important 
descendants of $v$, i.e., each node $v_i$ is an important node and there 
are no important nodes on the path between $v_i$ and $v$. 
There are $t\leq N^{1/f}$ highest important descendants  of $v$. 
For any $1\leq i \leq t$ and any $1\leq j \leq N_v$, 
 the data structure $E_v$ enables us to count elements with 
color $c\in C_{v_i}$ at positions $m\leq j$ of the array $A_v$. 
We can implement $E_v$ as follows. For any $1\leq i\leq t$, we 
store the positions of all elements in $A_v$ colored with a color 
from $C_{v_i}$ in a standard one-dimensional range counting data
 structure~\cite{Ch88}. 
Every such data structure uses $O(N_{v_i}\log N_{v})$ bits of space 
and answers queries in $O(\log N)$ time.  
Since $\sum_{i=1}^{t} N_{v_i}=N_v$, $E_v$ uses $O(N_v\log N_v)$ bits.

All nodes $v$  on the same level $l=i\floor{(1/f)\log N}$ 
contain $O(N)$ elements. Hence all data structures $REP_v$, $CNT_v$ 
and $E_v$ for important nodes $v$ situated on the same level use 
$O(N\log N)$ bits of space. Since important nodes are situated 
on a constant number of levels, the total space usage is $O(N\log N)$ 
bits.

The query answering procedure is similar to the one described in
 section~\ref{sec:nlognspace}, but only important nodes of $T$ are visited. 
The search starts at the root; we set $a_v=a$, $b_v=b$, and $i=t$, where 
$t$ is the number of the highest important 
descendants of $v$. \\
1. Let $a_{v_i}$ and $b_{v_i}$ denote the interval that corresponds 
to $[a,b]$ in $v_i$. If we know $a_v$ and $b_v$, we can find 
$a_{v_i}$ and $b_{v_i}$ using $E_v$. 
Then, we use the data structure in the node $v_i$ to compute the number of 
colors $m_{v_i}$ in $A_{v_i}[a_{v_i}..b_{v_i}]$ and the sum 
$r_i=\sum_{j=i}^{t} m_{v_j}$. \\
2. If $r_i< K$, we visit $v_i$, report all $m_{v_i}$ colors that occur 
 in $A_{v_i}[a_{v_i}..b_{v_i}]$ and proceed with the child $v_{i-1}$. \\
3. If $r_i\geq K$, we set $K= K-r_{i+1}$ and use the same procedure 
to report top $K$ colors that occur  in $A_{v_i}[a_{v_i}..b_{v_i}]$.  \\ 
The total number of visited nodes is $O(fN^{1/f})=O(N^{1/f})$. 
In every node we answer at most one color counting and one color reporting 
query; hence we obtain an unsorted list of top $K$ colors 
in $O(N^{1/f}\log N + K)$ time.
If $K<N^{1/f}$, we can sort $K$ colors by priorities 
in $O(N^{1/f}\log N)$ time. 
If $K\geq N^{1/f}$, we can use the radix sorting\footnote{ We assume that 
priorities of colors belong to the range $[1,O(N)]$. If this is not the case, 
then priorities can be replaced by their ranks.}
 and sort colors in $O(K)$ time.
Thus a query is answered in $O(N^{1/f}\log N + K)$ time.

Since  $O(N^{1/f'}\log N)=O(N^{1/f})$ for $f'>f$, we can substitute $f'>f$ 
in the above construction and obtain a data structure with 
$O(N)$ space and $O(N^{1/f})$ query cost for any constant $f$. 
\begin{lemma}\label{lemma:nspace}
For any constant $f$, there exists an $O(N\log N)$ bit data structure that 
supports top-$K$ color queries in $O(N^{1/f} + K)$ time.
\end{lemma}

\section{A Data Structure with $O(K)$ Query Time}
\label{sec:opttime}
In this section we will use the result of Lemma~\ref{lemma:nspace}
for $f=2$ as the starting point.  Although the reporting time
$O(\sqrt{N}+K)$ is very high, $O(N^{1/2}+K)=O(K)$ if
$K=\Omega(N^{1/2})$. Hence, the data structure of
Lemma~\ref{lemma:nspace} is optimal for $K=\Omega(\sqrt{N})$.
We can take care of the case when $K< \sqrt{N}$ colors must be
reported by explicitly storing the solutions for some pre-defined
queries and storing recursively defined data structures for subarrays.
We start by explaining the main ideas of our approach; a more detailed
description will be given later in this section.
 
{\bf Our Approach.} 
Lemma~\ref{lemma:nspace} enables us to answer top-$K$ queries in 
$O(K+\sqrt{N})=O(K)$ time when $K\geq \sqrt{N}$. 
Using the approach described below,  we can store the answers to top-$K$ 
queries for $K\leq \sqrt{N}$ and for a set of  intervals with  $O(\sqrt{N})$ 
endpoints using linear space.  
Let $J=\{\, i\floor{\sqrt{N}}\,\}$ and let $L(m,a,b)$ denote the set
of top $m$ colors in $A[a..b]$ sorted in the decreasing order by
priorities.  For every $i\in J$ and for every interval $[i-2^r,i]$ and
$[i,i+2^r]$, $r=1,2,\ldots,\log N$, we explicitly store the lists
$L(\sqrt{N},i-2^r,i)$ and $L(\sqrt{N},i, i+2^r)$.  For any interval
$[a,b]$, such that $a\in J$ and $b\in J$, we can represent $[a,b]$ as
a union of two (possibly intersecting) intervals $[a,a+2^x]$ and
$[b-2^x,b]$ for $x=\floor{\log(b-a)}$.  We can find the top
$K\leq \sqrt{N}$ colors in $A[a..b]$ by examining the first $K$ colors
in $L(\sqrt{N}, a,a+2^x)$ and $L(\sqrt{N},b-2^x,b)$ and reporting the
$K$ colors with highest priorities.  Hence, special queries on
intervals $A[a..b]$, where $a$ and $b$ are from $J$, can also be
answered in $O(K)$ time. 

We store the data structure of
Lemma~\ref{lemma:nspace} for each subarray $A[i_1..i_2]$, such that
$i_2$ follows $i_1$ in $J$.  Since each data structure for
$A[i_1..i_2]$ contains roughly $\sqrt{N}$ elements, it supports
queries in $O(N^{1/4}+ K)$ time. This query time is optimal for $K\geq
N^{1/4}$.  We thus obtained a data structure with optimal query time
for $K\geq N^{1/4}$: each interval $[a,b]$ can be represented as a
union of three intervals $[a,a_1]$, $[a_1,b_1]$ and $[b_1,b]$ such
that $a_1\in J$ and $b_1\in J$. We can find sorted lists of top-$K$
colors in all three intervals as described above; then, we can
traverse the lists and identify the top colors in $O(K)$ time.  

We can apply the same construction once again and obtain optimal query time
for $K\geq N^{1/8}$. If we apply the same idea $O(\log \log N)$
times, then we can support queries in optimal $O(K)$ time for an
arbitrary $K$. The precise description is given below.

{\bf Data Structure.}
Let $\rho(l)=(1/2)^l$ and  $\Delta=\log^2 N$.
We define the sets $J_1,J_2,\ldots,J_h$, where
$h=O(\log\log N)$, as $J_l=\{ i\cdot \floor{N^{\rho(l)}}\cdot\Delta\,|\, 
0\leq i \leq N^{1-\rho(l)}/\Delta \,\}$.  The last index $h$ is chosen so 
that $N^{\rho(l)}=\mathtt{const}$.  
For every $j\in J_l$, $1\leq l \leq h$, we store
$L(\ceil{N^{\rho(l)}},j-2^r,j)$ and $L(\ceil{N^{\rho(l)}},j,j+2^r)$
for $r=1,2,\ldots,\log N$.

For every subarray $A[j_1..j_2]$, such that $j_2$ follows $j_1$ in
$J_l$ and $1\leq l < h$, we store the data structure $R(l,j_1,j_2)$
implemented as in section~\ref{sec:linspace}.  For every subarray
$A[j_1..j_2]$, such that $j_2$ follows $j_1$ in $J_h$, we store the
data structure $R(h,j_1,j_2)$.  $R(h,j_1,j_2)$ is also implemented
according to Lemma~\ref{lemma:nspace}, but we set the constant $f=6$, 
so that $R(h,j_1,j_2)$ answers top-$K$ queries in $O((j_2-j_1)^{1/6}+K)$ 
time.
For every subarray $A[j_1..j_2]$, such that $j_2$ follows $j_1$ in
$J_h$ we also store a data structure $F(j_1,j_2)$ that supports
top-$K$ color queries in $O(K)$ time in the case when 
$K\leq \log^{1/3} N$.  This data structure will be
described later in this section.  Data structures
$R(\cdot,\cdot,\cdot)$ and $F(\cdot,\cdot)$ use modified sets of
colors.  Let $C(j_1,j_2)$ be the set of all colors that occur in
$A[j_1..j_2]$.  Let $\cM(j_1,j_2)$ denote the set in which every color
in $C(j_1,j_2)$ is replaced by the rank of its priority  in
$C(j_1,j_2)$: $\cM(j_1,j_2)=\{\, \prank(c,C(j_1,j_2) )\,|\,c\in
C(j_1,j_2)\,\}$, where $\prank(c,C)=|\{ c'\in C\,|\, p(c')\leq
p(c)\,\}$.  We store the data structure $R(l, j_1,j_2)$ or
$F(j_1,j_2)$ for the set of colors $\cM(j_1,j_2)$ and assume that the
priority of a color $p\in \cM(j_1,j_2)$ equals to $p$.  We also store a
table $Tbl(j_1,j_2)$ that enables us to find  the color $c\in
C$ that corresponds to a color $p\in \cM(j_1,j_2)$.  In this way we
guarantee that all colors and their priorities in $R(l,j_1,j_2)$ or
$F(j_1,j_2)$ belong to the range $[1,j_2-j_1+1]$.


{\bf Space Usage.}
We turn to the space analysis of our method. 
All lists $L(\cdot,\cdot,\cdot)$ use $o(N)$ bits: for each
$l$, there are $O(\frac{N}{N^{\rho(l)}\log^2 N})$ lists and each
list uses $O(N^{\rho(l)}\log N)$ bits. Hence, the total space usage
of all lists for a fixed $l$ is $O(N/\log N)$ bits.

Every table $Tbl(j_1,j_2)$ for a data structure $R(l,j_1,j_2)$ 
can be stored in $O(N^{\rho(l)}\log (N^{\rho(l)}))$ bits. 
For every color $p$, $1\leq p \leq |C(j_1,j_2)|$, the $p$-th 
entry contains a pointer to the color $c_p$, such that 
$\prank(c_p,C(j_1,j_2))=p$. The pointer to $c_p$ is the relative 
position of an element of color $c_p$ in $A[j_1..j_2]$. 
That is, $Tbl(j_1,j_2)[p]=t$ for some $t$ such that the color of $A[j_1+t]$
is $c_p$. Since $c_p\in C(j_1,j_2)$,  such $h$ always exists. 
Since $0\leq h \leq N^{\rho(l)}$, we need $O(N^{\rho(l)}\log (N^{\rho(l)}))$ 
bits to store the table. 

We can also store all data structures
 $R(\cdot,\cdot,\cdot)$ in $O(N\log N)$ bits: every data 
structure $R(l,j_1,j_2)$ contains $O(N^{\rho(l)})$ elements 
and colors of elements belong to the interval $[1,N^{\rho(l)}]$. 
Hence, by Lemma~\ref{lemma:nspace} we can store each 
$R(l,j_1,j_2)$ in $O(N^{\rho(l)}\log(N^{\rho(l)}))$ bits. 
Thus for a fixed value of $l$, all data structures $R(l,j_1,j_2)$ with 
tables $Tbl(j_1,j_2)$  use 
$O(N\rho(l)\log N)=\rho(l)O(N\log N)$ bits of space (the constant hidden in 
the big Oh notation does not depend on $l$). 
Since $\sum_{l=0}^{h} \rho(l)=O(1)$, all data structures 
$R(\cdot,\cdot,\cdot)$ use $O(N\log N)$ bits.  We will show below that 
all $F(\cdot,\cdot)$ also use $O(N\log N)$ bits. Thus the total space usage 
of our construction is $O(N\log N)$ bits.

{\bf Answering Queries.}
The procedure for reporting top $K$ colors in the range $[a,b]$
consists of the following steps.
In step 1, we identify the actual number of reported colors $K_q$: if 
$A[a..b]$ contains $K'$ distinct colors, then $K_q=\min(K',K)$. 
In step 2, we represent $[a,b]$ as a union of at most three intervals. 
Top $K_q$ colors in the middle interval $A[a_1..b_1]$ can be found using 
lists $L(\cdot,\cdot,\cdot)$, as explained in step 3. 
During steps 4-6 we find top $K$ colors in  the two other intervals, 
$A[a..a_1]$ and $A[b..b_1]$.\\
1. We check, whether the number of distinct colors in $[a,b]$ exceeds
$K$.  Using the data structure of Lemma~\ref{lemma:ram}, we report
colors in the interval $[a,b]$ until $K$ colors are reported or the
procedure stops. If the procedure stops when $K'<K$ colors are
reported, we set
$K_q=K'$. Otherwise, we set $K_q=K$.\\
2. We identify the largest value $l$, such that $N^{\rho(l)}\geq K_q$.
We can find $l$ by searching among $h=O(\log\log N)$ different
values in $O(1)$ time using the result of~\cite{FW94}.  Let
$a_1=\ceil{a/\ceil{N^{\rho(l)}}}$ and $b_1=\floor{b/\ceil{N^{\rho(l)}}}$.
\\
3. We identify the top $K_q$ colors in $[a_1,b_1]$ by computing
$x=\log(a_1-b_1)$ and examining the top $K_q$ colors in
$L(\ceil{N^{\rho(l)}},a_1,a_1+2^x)$ and $L(\ceil{N^\rho(l)},b_1-2^x,b_1)$. \\
4. If $l< h$, we use the data structure
$R(l,(a_1-1)\ceil{N^{\rho(l)}},a_1\ceil{N^{\rho(l)}})$ to identify the top 
$K_q$ colors in $A[a..a_1]$ in $O((N^{\rho(l)})^{1/2}+K_q)=O(N^{\rho(l+1)}+K_q)$
time.  Since $K_q>N^{\rho(l+1)}$, all colors in $A[a..a_1]$ are found
in $O(K_q)$ time.
The top $K_q$ colors in $A[b_1..b]$ are found in the same way.\\
5 If $l=h$ and $K_q\geq \log^{1/3} N$, we use the data structure
$R(h,(a_1-1)\ceil{N^{\rho(h)}},a_1\ceil{N^{\rho(h)}})$.  
Since $N^{\rho(h)}=O(\log^2
N )$, this takes $O((\log^2 N)^{1/6}+
K_q)=O(K_q)$ time by Lemma~\ref{lemma:nspace}. The top $K_q$ colors
in $A[b_1..b]$ are found
in the same way.\\
6. If $l=h$ and $K_q < \log^{1/3} N$, we use data
structures $F((a_1-1)\ceil{N^{\rho(h)}},a_1\ceil{N^{\rho(h)}})$ and
$F((b_1\ceil{N^{\rho(h)}},(b_1+1)\ceil{N^{\rho(h)}})$ to report top 
$K_q$ colors in $A[a..a_1]$ and $A[b_1..b]$.\\
7. When we know the top colors in $A[a..a_1]$, $A[a_1..b_1]$, and
$A[b_1..b]$, the top $K_q$ colors in $A[a..b]$ can be found in $O(K_q)$ time.\\

{\bf  Data Structure $F$.}
We  describe the data structure $F(j_1,j_2)$, where $j_1$
and $j_2$ are two consecutive indexes in $J_h$.
Since every color in $\cM(j_1,j_2)$ belongs to $[1,j_2-j_1+1]= [1,
O(\log^2 N)]$, every color in $\cM(j_1,j_2)$ can be specified with
$O(\log \log N)$ bits.  Let $V(j_1,j_2)= \{\,v_i \,\}$ for
$v_i=j_1+i\floor{\sqrt{\log N}}$ and $v_i\leq j_2$.  For every $v_i$ and for 
any $r$
such that $v_i+2^r\leq j_2$, we store the list $L(\ceil{\log^{1/3}
  N}, v_i,v_i+2^r)$; for every $v_i$ and for any $r$ such that
$v_i-2^r\geq j_1$, we store the list $L(\ceil{\log^{1/3}_2
  N},v_i-2^r,v_i)$.  Since there are $O(\log \log N)$ different
values of $r$ for each $v_i$, all lists $L(\cdot,\cdot,\cdot)$ use
$o((j_2-j_1)\log N)$ bits of space.  For any two consecutive indexes $v_i$ and
$v_{i+1}$ in $V(j_1,j_2)$, we store colors of all elements in
$A[v_i..v_{i+1}]$ in one machine word $W(v_i,v_{i+1})$.  Using one
look-up table of size $o(N)$ for all words stored in the data
structure and standard bit operations we can answer top-$K$ queries on
$A[v_i..v_{i+1}]$ in $O(K)$ time.

Any interval $[a,b]$ for $j_1 \leq a < b \leq j_2$ can be represented
as a union of intervals $[a,a_f]$, $[a_f,b_e]$, and $[b_e,b]$, 
where $a_f=\ceil{(a-j_1)/g }$,
$b_e=\floor{(b-j_1)/g }$ and $g=\ceil{\sqrt{\log N}}$.  We can
use lists $L(\ceil{\log^{1/3} N},a_f,a_f+2^x)$ and $L(\ceil{\log^{1/3} N},
b_e-2^x,b_e)$ for $x=\log(b_e-a_f)$ to find the top $K$
colors in $A[a_f..b_e]$.  We can find the top $K$ colors in
$A[a..a_f]$ and $A[b_e..b]$ in $O(K)$ time.  Finally, we can merge
the three lists and obtain the list of top $K$ colors in $O(K)$
time.

Thus we obtain
\begin{theorem}\label{theor:main1}
  There exists an $O(N\log N)$ bits data structure that
  supports top-$K$ color queries in $O(K)$ time.
\end{theorem}
The data structure described above can be constructed in 
$O(N\log N)$ time using 
the following algorithm.
Since all data structures $R(\cdot,\cdot,\cdot)$ contain $O(N\log \log N)$ 
elements, all  $R(\cdot,\cdot,\cdot)$ can be constructed in
 $O(N\log \log N)$ time. Data structures $F(\cdot,\cdot)$ can also be 
constructed in $O(N)$ time. 

We construct a data structure of Lemma~\ref{lemma:nlogn} in $O(N\log N)$ 
time and use it to generate lists $L(\cdot,\cdot,\cdot)$.
The total number of lists $L(\cdot,\cdot,\cdot)$ is $O(N\log N/\Delta )$ and 
the total number of elements in all 
 $L(\cdot,\cdot,\cdot)$ is $O((N/\Delta)\log N \log\log N)$.
Every list  $L(N^{\rho(l)},j_1,j_2)$ can be generated in  $O(\log^2 N + N^{\rho(l)})$ time. Hence, all lists  are constructed 
in  $O(N (\log^3 N/\Delta)+ (N/\Delta)\log N \log\log N)
=O(N (\log^3 N/\Delta))$ time. 
Thus the data structure of Theorem~\ref{theor:main1} can be constructed in 
$O(N\log N)$ time. 

The result of Theorem~\ref{theor:main1} can be also extended to the external 
memory model. We will show it in section~\ref{sec:extmem}

\section{An $O(N\log\sigma)$ Bit Data Structure}
\label{sec:optspace}
We can further improve the result of Theorem~\ref{theor:main1} 
in the case when $\sigma= o(N)$ and construct a data structure 
that uses $O(N\log\sigma)$ bits\footnote{In this section we assume that all colors 
are positive integers bounded by $\sigma^{O(1)}$. In the general case, our 
construction needs $O(\sigma\log m)$ additional bits, where $m$ is the maximal 
color}. 
In this section we assume w.l.o.g. that every color is an integer 
between $1$ and $\sigma$ (if this is not the case, we replace the color 
by the rank of its priority).

Our main idea is to divide  the array into chunks, 
so that the data structure for each chunk uses $O(\log \sigma)$ 
bits per element. We also need  the ``global'' data structure 
for searching in many chunks; this data structure 
contains $O(N/\log N)$ elements and therefore can be stored in $O(N)$ bits.
In the following description we will 
distinguish between two cases, $\sigma^2 \geq \log N$ and 
$\sigma^2 < \log N$. Although the data structures for both situations 
 are based on the same idea, for ease of description 
we discuss the two cases separately. 

{\bf Case 1: $\sigma^2 > \log N$.}
The array $A$ is divided into chunks $L_i$ so that each chunk 
contains $\ell=\sigma^3$ elements, $L_i=A[(i-1)\ell+1\,..\, i\ell]$.
We store the data structure of Theorem~\ref{theor:main1} for every chunk. 
Every such data structure uses $O(\ell\log\sigma)$ bits; all chunk 
data structures use $O(N\log\sigma)$ bits of space.

Besides that, we store a top level data structure $D^T$ for an auxiliary 
 array $A^T$. 
The array $A^T$ contains $\sigma$ entries for each chunk; 
the entries $A^T[(i-1)\sigma+1\,..\, i\sigma]$ contain information about colors 
that occur in the chunk $L_i$. 
If a color $c$ occurs in $L_i$, then $A^T[(i-1)\sigma+c]$ is colored 
with $c$. If $c$ does not occur in $L_i$, then $A^T[(i-1)\ell+c]$ is colored
with a dummy color $c_{\cD}$; we assume that the priority 
of $c_{\cD}$ is smaller than the priority of any other color. 
We  store the top-$K$ data structure of Theorem~\ref{theor:main1} 
for $A^T$. Since the total number of elements in $A^T$ is $O(N\sigma/\ell)=
O(N/\sigma^2)=O(N/\log N)$, both $A^T$ and $D^T$ use $O(N)$ bits.

If $a$ and $b$ belong to the same chunk, then we can answer the 
query $Q=[a,b]$ using the data structure for this chunk.
Otherwise, 
a query range $Q=[a,b]$ can be decomposed into $[a,a']$, $[a'+1,b']$, 
and $[b'+1,b]$ for $a'=\ceil{a/\ell}\ell$ and $b'=\floor{b/\ell}\ell$. 
Top $K$ colors in $A[a..a']$ and $A[b'+1..b]$ can be found using chunk 
data structures. Top $K$ colors in $A[a'+1\,..\,b']$ can be found 
using $D^T$: a color $c$ occurs in $A[a'+1\,..\,b']$ if and only if 
$c$ occurs in $A^T[\ceil{a/\ell}\sigma+1 \,..\,  \floor{b/\ell}\sigma]$.  
Hence, we can find the top $K$ colors in $A[a'+1\,..\, b']$ by identifying 
the top $K$ colors in $A^T[\ceil{a/\ell}\sigma+1 \,..\,  \floor{b/\ell}\sigma]$ 
and discarding the dummy color $c_{\cD}$ if necessary. 
Since all three lists of top $K$ colors are sorted by priorities, 
we can merge them and obtain top $K$ colors in $O(K)$ time.

{\bf Case 2: $\sigma^2 < \log N$.} 
We use the same construction as for the case $\sigma^2\geq \log N$: 
the array $A$  is divided into 
chunks and  there is a data structure for a each chunk. There is also 
 a data structure $D^T$ for the array $A^T$ defined as above.   
But now  each chunk $L_i$ consists of $\sigma^2\floor{\log N}$ elements. 
It remains to describe the new chunk data structure. 

Every chunk $L_i$ is divided into $O(\sigma^2\floor{\log\sigma})$ pieces 
$\cP_j$ and 
each piece consists of $\ceil{\log_{\sigma} N}$ elements. 
The array $L_i^T$ contains $\sigma$ entries for each piece. 
If the $j$-th piece $L_i[j\ceil{\log_{\sigma} N} +1\,..\,  (j+1)\ceil{\log_{\sigma} N}]$ contains 
an element 
of color $c$, then the color of $L_i^T[j\sigma+c]$ is $c$; otherwise 
the color of $L_i^T[j\sigma+c]$ is the dummy color  $c_{\cD}$. 
The top-$K$ color data structure for $L_i^T$ needs $O(\sigma^3\log^2\sigma)$ 
bits of space.

Every piece $\cP_j$ contains $O(\log_{\sigma} N)$ elements. Since the color 
of each element can be stored in $O(\log \sigma)$ bits, each piece fits 
into $O(1)$ words of $\log N$ bits. We can answer top-$K$ queries on 
each piece using a pre-computed table $\cT$ of size $o(N)$. The table 
$\cT$ contains 
information about all possible sequences $\seq_l$ of  
$\floor{\log_{\sigma} N/2}$   colors.
For every sequence  $\seq_l$ of $\floor{\log_{\sigma} N/2}$ colors and for any 
$1\leq x_1\leq x_2\leq \floor{\log_{\sigma} N/2}$, we store 
all colors that occur between $x_1$ and $x_2$ in $\seq_l$ 
sorted in decreasing order by their priorities. 
Since there are $O(\sqrt{N})$ sequences $\seq_l$, such a table uses 
$O(\sqrt{N}\log_{\sigma}^2 N\log^2\sigma) =o(N)$ bits of space.
We can decompose a piece into a constant number of sequences $\seq_f$ of 
$\floor{\log_{\sigma} N/2}$ colors; for every  
sequence $\seq_f$, we can look-up the corresponding entry 
in the table $\cT$ and identify (at most) top-$K$ colors 
between any two positions of $\seq_f$ in $O(1)$ time. 
Hence, we can answer a top-$K$ query on a 
piece $\cP_j$ in $O(K)$ time. 

We described in the first part of  section~\ref{sec:optspace} how a 
top-$K$ query 
to an array $A$ can be decomposed into two queries on chunks 
$L_{i_1}$ and $L_{i_2}$ and one query to a data structure for $A^T$.
In the same way,
the top-$K$ color query on a chunk $L_i$ can be answered by answering 
two queries on some pieces $\cP_{j_1}$ and $\cP_{j_2}$ and one query to 
a data structure for $L_i^T$. Hence, a top-$K$ query on a chunk can be 
answered in $O(K)$ time. 

Thus we obtain the following result
\begin{theorem}\label{theor:mainspace}
Let $\sigma$ be the number of different colors.
There exists an $O(N\log \sigma)$ bits data structure that
  supports top-$K$ color queries in $O(K)$ time.
\end{theorem}
The construction time of this data structure is $O(N\log \sigma)$; 
this can be shown in the same way as for the data structure of 
Theorem~\ref{theor:main1}.  

\section{An External Memory Data Structure}
\label{sec:extmem}
In the external memory model~\cite{AV88}, the data is stored on a disk in 
blocks and   all computations are performed on data stored in the main memory 
of size $M$. A block consists of $B$ words of $\log N$ bits. 
Using one I/O operation, we can read a block of data from disk 
or write it into disk. The  time complexity of external memory
 algorithms is measured in I/O operations and the space usage is measured 
in blocks. Our data structure for top-$K$ queries can be also extended 
to the external memory data structure that uses $O((N/B)\log\log N)$ blocks of space and answers queries in $O((K/B)\log_B K)$ I/O operations (if the block 
size $B$ is not very small).

In this section we will 
distinguish between two variants of the top-$K$ color problem. 
 In the sorted top-$K$ color problem, $K$ colors with highest priorities 
 must be reported in the descending order of their priorities.
 In the unsorted top-$K$ color problem, $K$ colors with highest priorities 
 are reported in an arbitrary order. 
Sometimes we will specify the space usage of our data structures in bits. 
 We observe, however, that if we describe an external memory 
  data structure that uses 
 $O(s(N))$ bits, then this data structure 
 can be packed into $O((s(N)/(B\log N)))$ block of space. 

The data structure of Lemma~\ref{lemma:nlogn} can be straightforwardly extended to the external memory model. The only difference is that data structures $REP_v$ and $CNT_v$ are implemented as explained in Lemma~\ref{lemma:ext}.
\begin{lemma}\label{lemma:nlognext}
There exists an $O(N\log^2 N)$ bits external memory data structure that 
reports top-$K$ colors in $O(\log^2 N + K/B)$ time. 
\end{lemma}

The space usage can be further reduced to $O(N\log N)$ bits.
In external memory model, we can obtain an unsorted list of top $K$ colors 
in the same way as explained in Lemma~\ref{lemma:nspace}
 (counting and reporting data structures are 
replaced with their external memory counterparts). 
However, there is no external memory equivalent of the radix sort. 
Hence, we need $O((K/B)\log_B K)$ I/Os to sort the $K$ colors. 
\begin{lemma}\label{lemma:nspaceext}
For any constant $f$, there exists an $O(N\log N)$ bits data structure that 
supports unsorted top-$K$ color queries in $O(N^{1/f} + K/B)$ I/Os.\\
For any constant $f$, there exists an $O(N\log N)$ bits data structure that 
supports top-$K$ color queries in $O(N^{1/f} + (K/B)\log_B K)$ I/Os.\\
\end{lemma}

We can extend the  result of Theorem~\ref{theor:main1} 
to the external memory model. The only
major difference is that all data structures $R(\cdot,\cdot,\cdot)$ 
and $F(\cdot,\cdot)$
use the same set of colors $C$. The data structures $R(\cdot,\cdot,\cdot)$ 
are implemented using  Lemma~\ref{lemma:nspaceext}. 
We can implement each data structure $F(j_1,j_2)$ 
 using Lemma~\ref{lemma:nlognext};
since $F(j_1,j_2)$ contains $m= O(\log^2 N)$ elements, it can be implemented 
with $O(m\log N(\log\log N))$ bits, so that queries are answered in 
$O((\log\log N)^2 + K/B)$ time. 
On the other hand, if $B\geq \log^2 N$, then we can 
pack all elements of $F(\cdot,\cdot)$  into one block of space and thus answer 
the top-$K$ color queries on $F(\cdot,\cdot)$ in $O(K/B)$ I/Os. 
\begin{lemma}\label{lemma:boot}
  There exists an external memory data structure that uses
  $O((N/B)\log\log N)$ blocks of space and supports unsorted
  top-$K$ color
  queries in $O((\log\log N)^2+ K/B)$ I/O operations. \\
  There exists an external memory data structure that uses
  $O((N/B)\log\log N)$ blocks of space and supports top-$K$ color
  queries in $O((\log\log N)^2 + (K/B)\log_B K)$ I/O operations.
\end{lemma}
We can further improve the query cost by bootstrapping: we use 
Lemma~\ref{lemma:boot} to implement each data structure $F(j_1,j_2)$. 
Since $F(j_1,j_2)$ contains $O(\log^2 N)$ elements, unsorted  
queries are supported in $O((\log^{(3)}_2 N)^2 + K/B)$ I/O operations. 
Here $\log^{(t)} N$ is defined as $\log^{(t)} N=\log(\log^{(t-1)} N)$
for $t>1$ and $\log^{(1)}N=\log N$. 
The improved  data structure also  supports queries in $O(K/B)$ I/Os 
if $B>(\log\log N)^2$. 
If we apply the same idea $t$ times, we obtain the following result. 
\begin{theorem}\label{theor:main2}
  Let $t$ be an arbitrary positive integer.
  There exists an external memory data structure that uses
  $O((N/B)\log\log N)$ blocks of space and supports unsorted
  top-$K$ color
  queries in $O((\log^{(t)} N)^2 + K/B)$ I/O operations. 
  If $B> (\log^{(t-1)} N)^2$, then queries are supported in 
  $O(K/B)$ I/Os.\\
  There exists an external memory data structure that uses
  $O((N/B)\log\log N)$ blocks of space and supports top-$K$ color
  queries in $O((\log^{(t)} N)^2 + (K/B)\log_B K)$ I/O operations.
  If $B> (\log^{(t-1)} N)^2$, then queries are supported in 
  $O((K/B)\log_B K)$ I/Os.
\end{theorem}

\section{Online Queries}
\label{sec:online}
We assumed in the above description that the data structure must report 
$K$ top colors in the query interval, and the value of $K$ is specified 
with the query. The same results are also valid in the online reporting 
scenario: the data structure reports top colors from a specified interval 
until the user terminates the reporting procedure or all colors in the 
interval are reported. 

It suffices to apply the following trick, described in 
e.g.~\cite{BFGO09}.
Let $K_i=2^i$. The reporting procedure consists of stages indexed by 
$i=-1,0,1,2,\ldots$. 
During the $i$-th stage we: 
(i) use the data structure of Theorem~\ref{theor:main1} or Theorem~\ref{theor:main2} to generate the sorted list $\cL_{i+1}$ of the top $2K_{i+1}-1$ colors, 
(ii) remove the first $K_{i+1}-1$ elements from $\cL_{i+1}$, and  
(iii) if $i\geq 0$, output  the colors from $\cL_i$.

By Theorem~\ref{theor:main1}, we can 
find the top $2K_{i+1}-1$ colors in $O(2K_{i+1})=O(K_i)$ time.
Hence, steps (i) and (ii) above take $O(K_i)$ time. 
The list $\cL_0$ is produced during the initial $(-1)$-th stage  
in $O(1)$ time. 
For $i\geq 0$,  we interleave steps (i) 
and (ii) with the step (iii): every time when we output one element 
of $\cL_i$, we spend $O(1)$ time on steps (ii) and (iii). Hence, the list 
$\cL_{i+1}$ is known when the stage $i$ is completed.

\section{Document Retrieval}
\label{sec:docret}
{\bf Preliminaries.}
The generalized suffix tree for documents $d_1,\ldots, d_s$ is the compact 
trie that contains all suffixes of the string 
$d_1\$_1\ldots d_{s-1}\$_{s-1}d_s$, where $\$_1<\$_2<\ldots<\$_s$ are 
additional dummy symbols. The path of a node $v$ is the string 
obtained by concatenating all edge labels on the path from the root to $v$. 
The locus of a pattern $P$ is the highest node $v$, such that $P$ is the 
prefix of the path of $v$. All occurrences of $P$ correspond to the leaf
descendants of the locus of $v$. 
The locus of $P$ can be found in $O(|P|)$ time. We refer to e.g.,~\cite{HSV09,VM07} and references therein for a more detailed description.

{\bf Ranked Document Listing.} 
We store all leaves of the generalized 
suffix tree in an array $A$. We set the color of the $i$-th element to 
$c_j$ if and only if the suffix corresponding to the $i$-th leaf $l_i$ 
belongs to the document $d_j$; the priority of the color $c_j$ equals to 
the priority of the document $d_j$, $p(c_j)=p(d_j)$. 
In every internal node of the suffix array, we store the maximal and minimal 
index of its leaf descendants. 

Now the ranked document listing query can be answered as follows. 
We identify the locus $v$ of a query pattern $P$ in $O(|P|)$ time. 
Let $\min_v$ and $\max_v$ be the minimal and maximal indexes of the leaf 
descendants of $v$. Reporting top-$K$ colors that occur in $A[\min_v..\max_v]$ 
is equivalent to reporting $K$ most highly ranked documents that contain 
$P$ in the reverse order of their ranks. 
Hence, we can answer a ranked document listing query in $O(|P|+K)$ 
time. Additional space needed by our data structure is $O(N\log s)$ bits, 
where $N$ is the total length of all documents.
\begin{corollary}
There exists an $O(N\log N)$ bits data structure that 
supports ranked document listing queries in $O(|P|+K)$ time.
\end{corollary}
{\bf Ranked $t$-Mine Problem.}
Muthukrishnan~\cite{M02} showed how we can identify all documents that 
contain at least $t$ occurrences of a pattern $P$ by reporting 
all colors in an array $A^t[a_p..b_p]$ that contains $O(N/t)$ elements.
That is, for any pattern $P$ we can identify indexes $a_p$ and $b_p$ 
in $O(|P|)$ time, so that a document contains at least $t$ occurrences of  $P$ 
if and only if the corresponding color occurs in $A^t[a_p..b_p]$ at 
least once.   Details about the array $A^t$ can be found in~\cite{M02}.
All  $A^t$ for all possible values of $t$ contain 
$O(\sum_t (N/t) )= O(N\log N)$ elements. 

We store the data structure $D^t$ of Theorem~\ref{theor:mainspace} for 
each $A^t$. Using this data structure, we can report top-$K$ colors 
in $A^t[a_p..b_p]$. Obviously, this is equivalent to reporting $K$ 
 most highly ranked documents that contain $t$ occurrences of a pattern $P$. 
Each $D^t$ can be stored in $O((N/t)\log s)$ bits, where $s$ is the number of 
documents. Hence, all data structures $D^t$ use $O(N\log s)$ words of 
$\log N$ bits.
\begin{corollary}
There exists a data structure that  uses 
$O(N\log s)$ words, where $s$ is the 
number of documents, and 
supports ranked $t$-mine  queries in $O(|P|+K)$ time.
\end{corollary}
{\bf Most Relevant Documents Problem.} 
Recently, Hon et al.~\cite{HSV09} developed a framework for reporting $K$ most 
relevant documents. In addition to reporting $K$ most highly ranked documents,
the structure of~\cite{HSV09} enables us to report $K$ documents with highest 
scores, so that  the score depends on both the document $d$ and the 
pattern $P$. Combining their approach with our data structure, we can 
slightly improve their results.

Using the construction of~\cite{HSV09}, the problem of reporting 
$K$ most highly scored documents with respect to a metric $\rel(d,P)$  is 
reduced to a problem of reporting $K$ highest 
values that occur in $|P|$ non-overlapping ranges
$\cA[a_1..b_1]$, $\cA[a_2..b_2]$,$\ldots$, $\cA[a_{|P|}..b_{|P|}]$ of the 
array $\cA$. The array $\cA$ of size $O(N)$ contains document identifiers, 
and every document occurs in 
$\cA[a_1..b_1]$, $\cA[a_2..b_2]$,$\ldots$, $\cA[a_{|P|}..b_{|P|}]$ at most 
once. We refer to~\cite{HSV09} for the description of their data structure.

Our improvement is based on storing the array $\cA$  in the data structure of 
Theorem~\ref{theor:mainspace}. At the beginning of the search procedure, 
we identify the maximum element in every range $\cA[a_i..b_i]$ and 
store them in a heap data structure ( if $|P|>K$, then we store only $K$ 
largest elements in the heap). Then, we repeat the following 
steps $K$ times: (i) we extract the maximum element from the heap
and add it at the end of our  list of top documents
(ii) if the element belongs to the range $\cA[a_j..b_j]$, we obtain 
the next highest value in $\cA[a_j..b_j]$ and add it to the heap. 
Since the heap contains $|P|$ elements, extracting the maximum element 
from the heap and inserting a new element into the heap take
 $O(\log |P|)$ time. If $|P|=\log^{O(1)}N$, then heap 
operations can be implemented in $O(1)$ time~\cite{FW94,T99}.
The data structure of~\cite{FW94} uses 
multiplications or other time-consuming operations. 
In our case, however,  all elements stored in the heap are bounded by 
$N$. Hence, we can replace each time consuming operation by $O(1)$ 
 bit operations and look-ups in a table of size $o(N)$ that can be initialized 
in $o(N)$ time.
As explained in section~\ref{sec:online}, finding the next largest element in 
the range takes $O(1)$ time.
Thus the procedure takes $O(|P|+K\log|P|)$ time or $O(|P|+K)$ time if 
$|P|=\log^{O(1)}N$.
\begin{corollary}
There exists an $O(N\log N)$ bits data structure that 
supports most relevant documents queries in $O(|P|+K\log(|P|) )$ time.
If $|P|=\log^{O(1)} N$, then queries can be supported in $O(|P|+K)$ time.
\end{corollary}

\section{Conclusions}
In this paper we described a data structure for top-$K$ color reporting 
with optimal query time. The worst-case space usage of our data structure 
is also optimal. This result allows us to report most highly ranked 
documents that contain a query pattern $P$ 
in optimal time using optimal worst-case space. 
While the recent compressed data structure of Hon et al.~\cite{HSV09} 
uses less space it needs $O(\log^{3+\eps} N)$ time to report each document. 
It is an interesting open question, whether we can construct 
a compressed data structure that requires significantly less time to report 
every document.

\end{document}